%% file: main.tex
\documentclass[letterpaper]{article}

\usepackage[margin=1in]{geometry}

\input{preamble}

\input{title}

\begin{document}
	
	\maketitle
	
	{\input{abstract}}
	
	\newpage
	
	\input{introduction}

	\input{preliminaries}

	\input{body}

	\input{discussion}

	\bibliography{reference}
	
	\appendix
	
	\input{appendix}
	
\end{document}

%% file: preamble.tex
\usepackage{amsmath, amssymb, amsthm}
\usepackage[section]{algorithm}
\usepackage{algpseudocode}
\usepackage{color} 
\usepackage{mathrsfs} 
\usepackage{url}
\usepackage{comment}

\newcommand{\todo}[1]{{\color{red} TODO: #1}}

\newcommand{\nosection}[1]{\vspace{4pt}\noindent\textbf{#1. }}

\newcommand{\bs}{\boldsymbol}

\newcommand{\probs}{\mathrm{Pr}}
\newcommand{\prob}[1]{\probs \left( {#1} \right)}
\newcommand{\entps}{\mathrm{H}}
\newcommand{\entp}[1]{ \entps \left( {#1} \right) }
\newcommand{\renyi}[2]{ \entps_{#1} \left( {#2} \right) }
\newcommand{\minfos}{\mathrm{I}}
\newcommand{\minfo}[2]{ \minfos \left( {#1} ; {#2} \right) }
\newcommand{\expes}{\mathrm{E}}
\newcommand{\expe}[1]{ \expes \left( {#1} \right) }

\newcommand{\LOCAL}{\mathsf{LOCAL}}
\newcommand{\execs}{\eta} 
\newcommand{\exec}[3]{ \execs^{#1}_{#2} ( {#3} )} 
\newcommand{\hists}{h} 
\newcommand{\hist}[2]{\hists^{#1}_{#2}} 
\newcommand{\pref}[2]{(#1)_{ [ {#2} ]}} 
\newcommand{\leaks}{\Phi} 
\newcommand{\leak}[3]{ \leaks^{#1}_{#2} ( {#3} ) } 
\newcommand{\leacs}{\Psi} 
\newcommand{\leac}[3]{ \leacs^{#1}_{#2} ( {#3} ) } 
\newcommand{\wcoms}{\alpha} 
\newcommand{\wcom}[2]{ \wcoms_{#1} ( {#2} )} 
\newcommand{\acoms}{\beta} 
\newcommand{\acom}[2]{ \acoms_{#1} ( {#2} )} 
\newcommand{\filts}{\delta} 
\newcommand{\filt}[1]{ \filts ( {#1} ) } 
\newcommand{\card}[1]{ | {#1} | } 
\newcommand{\actes}{\mathscr{A}} 
\newcommand{\acte}[1]{ \actes ( {#1} ) } 
\newcommand{\discs}{d^*} 
\newcommand{\disc}[3]{\discs_{#1}({#2};{#3})}
\newcommand{\depts}{d} 
\newcommand{\dept}[2]{\depts_{#1}(#2)} 

\newtheorem{thm}{Theorem}[section]
\newtheorem{defn}[thm]{Definition}
\newtheorem{coro}[thm]{Corollary}
\newtheorem{lemma}[thm]{Lemma}
\newtheorem{prop}[thm]{Proposition}

\algnewcommand\algorithmicupon{\textbf{upon}}
\algdef{SE}[UPON]{Upon}{EndUpon}[1]{\algorithmicupon\ #1\ \algorithmicdo}{\algorithmicend\ \algorithmicupon}

%% file: title.tex
\title{A Rudimentary Model for Low-Latency Anonymous Communication Systems}

\author{Yufan Zheng \\ Peking University \\ {\small \url{yufan.zheng@pku.edu.cn}}}

\date{}

%% file: abstract.tex
\abstract{
	In this paper we present a rudimentary model for low-latency anonymous communication systems.
	Specifically, we study distributed OR algorithm as an abstract of the system.
	Based on our model, we give several satisfactory lower bounds of anonymity leakage of a deterministic OR algorithm.
	Some of them reveal a trade-off between anonymity and communication complexity.
	For the randomized OR algorithm, we only give a relatively trivial but possibly tight lower bound when leaving out communication complexity.
	And we find the relationship between our model and some open case in the study of secret sharing scheme, if considering communication complexity.
}

%% file: introduction.tex
\section{Introduction}
An \emph{anonymous communication system}, as a concept introduced by Chaum \cite{chaum1981untraceable} in 1981, hides the identities of nodes or participants when they communicate.
In recent years, it has gained more and more popularity because of various motivations, e.g., threat of pervasive surveillance or personal privacy preferences.
The best-known deployed design, Tor \cite{dingledine2004tor}, is used daily by about three million people \cite{metricstor}.
For a survey on designing, developing, and deploying systems for anonymous communication, see \cite{shirazi2016survey}.

Anonymous communication systems can be classified into low-latency and high-latency systems. Generally speaking, high-latency designs, such as Mixminion \cite{danezis2003mixminion} and Dissent \cite{corrigan2010dissent}, are secure even against a powerful global passive adversary, at the cost of giving up real-time communication.
In contrast, low-latency designs are more vulnerable to attacks, resulting in a certain degree of anonymity leakage.
In this paper we will focus on low-latency systems.
To be exact, our interest is how secure a low-latency anonymous communication system could be.
In other words, we wonder the \emph{theoretical bound} of the performance of systems.

There are already various formalizations of anonymous communication systems such as AnoA \cite{backes2013anoa} and UC \cite{canetti2001universally}, which are complete frameworks for analyzing anonymous communication systems.
However, these frameworks are too complicated for our purpose, though they are powerful when dealing with a concrete system or protocol.
Instead, we will try to build a model as simple as we can, so that we will be able to derive some meaningful result.
The most important simplification is that we do \emph{not} consider cryptography in our model, although a variant of model will be introduced in Section \ref{subsec:deterministicc}, which is involved with cryptography to some extent.
Another simplification is that we only consider a brief moment of the system, without regarding some complex dynamics of it, e.g., node churn in a peer-to-peer system\footnote{In such a system, nodes participate only for a short period of time, and the topology changes quickly.}.

We can assume that the low-latency anonymous communication system we will study is message-based, which means that the fundamental function of the system is to transmit messages among its nodes, where every receiver of the message would not obtain any information about the sender.
This assumption is reasonable, since all types of anonymous communication systems can be reduced to message-based ones.

\subsection{Modeling}
Though a strict mathematical description of our model will be given in Section \ref{sec:preliminaries}, it is necessary to briefly explain in advance how we obtain the model and why we accept the model. 

\nosection{Low-Latency Anonymous Communication Systems}
In order to build a general model for low-latency anonymous communication systems, we should grasp the most essential property of them.
Consider a low-latency anonymous communication system with delay at most $\Delta$.
By correctness of the system, we can assert that each node in the system can tell whether at least one node has sent some message to it in the last $\Delta$ time.
In other words, suppose the network topology does not change in the last $\Delta$ time, then let $V$ be the collection of all nodes in the system, $t$ be a node in $V$, and $I(v) = \bs{1}\{ v \text{ sent some message to } t \text{ in the last } \Delta \text{ time} \}$.\footnote{Symbol $\bs{1}\{P\}$ takes value $1$ if proposition $P$ is true, or $0$ if it is false.}
We have
\begin{prop}
	The low-latency anonymous communication system solves a distributed OR problem.
	To be exact, node $t$ can give the result of $\bigvee_{v \in V} I(v)$.\footnote{Surely it is feasible to replace $\bigvee_{v \in V} I(v)$ by $\sum_{v \in V} I(v)$ or something else for a stronger result, but it suffices for this paper to use the former.}
\end{prop}
The property above allows us to only focus on a simple distributed algorithm, if we can obtain some meaningful result after doing so.
Indeed, we will soon see that it does very well.
To determine the precise characteristic of the algorithm we will study, we should take the threat model into account.

\nosection{Threat Model and Anonymity Metric}
Generally speaking, an adversary is supposed to partly break the anonymity of the system, by identifying the sender of some message or at least revealing some information about it.
There are many possible cases which may happen in the system.
Among them, we are interested in the case the adversary knows that a node $t \in V$ recieved some message just now and that in the last $\Delta$ time there was only exactly one message transmitted by the system.
There are two reasons for our interest.
First, apparently the model will be simpler in this case.
Second, this is a plausible worst case, since more messages transmitted meanwhile bring more difficulty to the adversary.
We say that it is plausible because the system has relatively low latency $\Delta$, and the adversary can know that a node $t \in V$ receieved some message in many circumstances, e.g., when sender of the message is browsing the website controlled by the adversary, via the proxy built on the system.
To measure anonymity, we will use the classical entropy-based anonymity metric \cite{diaz2002towards, serjantov2002towards}.
The remaining issue is to determine the adversary's ability.
In this paper, we mainly consider an adversary who is external, passive and static, in terms of the standard terminology introduced by Raymond \cite{raymond2001traffic}.
However, it is possible to consider other types of adversaries by making some changes to our model, which will be discussed in Section \ref{subsec:otherissues}.

\nosection{Computational Model}
We will choose a synchronous model for distributed computing, which has two advantages.
First, although most of distributed systems are asynchronous in the real world, studying synchronous model is still a good start because it is relatively simpler.
Second, an adversary in synchronous model is stronger than one in asynchronous model, because synchronicity eliminates some uncertainties for the adversary.
And unfortunately in the real-world case, a system often behaves more like a synchronous one under existing attack techniques, e.g., end-to-end timing analysis \cite{levine2004timing, bauer2007low, zhu2010correlation}.

\subsection{Contribution}
In this paper we present a rudimentary model to analyze performances of low-latency anonymous communication systems from theoretical aspect. 
To be specific, we study \emph{distributed OR algorithm} as an abstract of real-world low-latency anonymous communication systems.
By introducing an external, passive and static adversary, we define \emph{anonymity leakage} of an OR algorithm, which is roughly consistent with the classical metric defined in practice.

After, we obtain some elementary results with our model.
Then, for deterministic OR algorithm, we give several lower bounds of anonymity leakage.
Theorem \ref{thm:dense} and Theorem \ref{thm:densec} give us general bounds, whatever the network topology is.
Theorem \ref{thm:sparse} and Corollary \ref{coro:sparsec} give us stronger bounds if the network topology can be specified by a relatively sparse graph, which usually corresponds to the case a real-world system is peer-to-peer with restricted topologies, e.g., Tarzan \cite{freedman2002tarzan} and ShadowWalker \cite{mittal2009shadowwalker}. 
Among them, Corollary \ref{coro:sparsec} and Theorem \ref{thm:densec} show a trade-off between anonymity and communication complexity.
For randomized case, we give a relatively trivial but possibly tight bound in Corollary \ref{coro:rcase}.
Unlike in the deterministic case, we find out the relationship between our model and some open case in the study of \emph{secret sharing scheme}, if considering communication complexity.

Finally, we discuss about the tightness issue on our bounds and something else related to our subject.

%% file: preliminaries.tex
\section{Preliminaries} \label{sec:preliminaries}

\subsection{Computational Model}
Throughout this paper, we study the standard $\LOCAL$ model \cite{linial1992locality, peleg2000distributed} for distributed computing.
It is also called \emph{synchronous network model} \cite{lynch1996distributed} sometimes.
We will just make a rough introduction here, and formal definition of the model can be found in \cite{peleg2000distributed}.

\nosection{Basic Concept}
In the $\LOCAL$ model, processors and bidirectional communication channels between them are represented by nodes and edges of a connected undirected {\it underlying graph} $G = (V,E)$, respectively.
With replicas of an {\it algorithm} $\Pi$ running at all nodes, the computation proceeds in synchronized {\it rounds}. 
In each round, every node $v \in V$ can receive messages sent by its neighbors in the previous round, do some local computation, and send at most one message to each its neighbor separately.
There is no restrictions on local computations, so we can simply regard every replica of $\Pi$ residing at $v \in V$, called \emph{process} $v$, as a (possibly infinite) state machine, which is more powerful than a Turing machine \cite{turing1937computable}.
For convenience, we require any algorithm $\Pi$ to be fundamentally {\it deterministic} in this paper.

\nosection{Initial Input, Execution and History}
An \emph{initial input} $I : V \to \mathcal{I}$ states what information each node $v \in V$ would get at the very beginning of computation, where $\mathcal{I}$ stands for the collection of possible inputs for a node.
An {\it execution} governed by algorithm $\Pi$ on underlying graph $G = (V, E)$ with initial input $I$ is denoted $\exec{\Pi}{G}{I}$.
For each edge $e = (u,v) \in E$ we define {\it history} 
\[
	\hist{\execs}{e} = \hist{\execs}{(u,v)} = (\hist{\execs}{u \to v}, \hist{\execs}{v \to u}) \in \mathcal{M}^{\omega} \times \mathcal{M}^{\omega}
\]
as the communication history of edge $e$ under execution $\execs = \exec{\Pi}{G}{I}$, which composed of two components for two direction of the communication, where $\mathcal{M}$ stands for the collection of possible messages (including empty message denoted $\lambda$) sent by some node through an edge within a round, and $\omega$ is the smallest infinite ordinal. 
To be specific, $\hist{\execs}{u \to v} = (m_1, m_2, \dots)$ indicates that $u$ sent a message $m_k$ to $v$ in round $k$ under $\execs$ for all $k \in \mathbb{Z}^+$. 
We also define \emph{prefix} of a history by 
\[
	\pref{\hist{\execs}{e}}{k} = (\pref{\hist{\execs}{u \to v}}{k}, \pref{\hist{\execs}{v \to u}}{k}) \mathrm{\ and\ } \pref{\hist{\execs}{u \to v}}{k} = (m_1, \dots, m_k)
\]
where $\hist{\execs}{u \to v} = (m_1, m_2, \dots)$ and $k \in \mathbb{N}$.

\nosection{Randomness}
In order to analyze a randomized algorithm, we allow every node $v \in V$ to acquire some extra information at the very beginning of computation from the \emph{random input} $R: V \to \{0,1\}^\omega$. 
Hence, each node $v \in V$ would hold an infinite binary string assigned by $R$.
An execution governed by algorithm $\Pi$ on underlying graph $G$ with initial input $I$ and random input $R$ is denoted $\exec{\Pi}{G}{I,R}$. 
If we let a random variable $\bs{R}$ satisfy that every bit of $\bs{R}(v)$ for every $v \in V$ are i.i.d. and uniformly distributed on $\{0,1\}$, we can achieve randomness by studying $\exec{\Pi}{G}{I,\bs{R}}$.

\nosection{Conventions}
Recall that we have assumed any algorithm $\Pi$ in this paper to be deterministic. 
So in the following text if we say that $\Pi$ is a randomized algorithm, we actually mean that $\Pi$ takes random input $R$.
We omit ``connected'' and ``undirected'' when we refer to an underlying graph $G$ by default.
Throughout this paper, we always let every node $v \in V$ be initially aware of the underlying graph $G$ and its ID ``$v$''.
Namely, complete topological knowledge of the network is always assumed to be included in $I(v)$ for every $v \in V$.
However, for simplicity, we will abuse notation by leaving out these information when defining or using the symbol $I$.
For random input $\bs{R}$ written as a random variable, we always follow the definition in the previous paragraph.
Hence we will no longer define $\bs{R}$ when referring to it in the following text.
Furthermore, note that a deterministic algorithm is randomized for sure.
Therefore without loss of generality we only study the randomized case when an algorithm is not specified to be whether randomized or deterministic.

By the nature of deterministic state machine we have
\begin{thm} \label{thm:consistency}
	Let $G = (V, E)$ be an underlying graph, $v_0 \in V$ be a node, $\Pi$ be an algorithm and $I, I'$ be inputs satisfying $I(v_0) = I'(v_0)$.
	Denote $\execs = \exec{\Pi}{G}{I}$ and $\execs' = \exec{\Pi}{G}{I'}$.
	If $\pref{\hist{\execs}{v \to v_0}}{k} = \pref{\hist{\execs'}{v \to v_0}}{k}$ for all $v \in V$ adjacent to $v_0$ in $G$, the state of process $v_0$ under $\execs$ would be identical to that under $\execs'$ from round $1$ to round $k+1$. 
	In addition, $\pref{\hist{\execs}{v_0 \to v}}{k+1} = \pref{\hist{\execs'}{v_0 \to v}}{k+1}$ holds for all $v \in V$ adjacent to $v$ in $G$. 
	Also, there is a same conclusion for $\execs = \exec{\Pi}{G}{I,R}$ and $\execs' = \exec{\Pi}{G}{I',R'}$ where $I(v_0) = I'(v_0)$ and $R(v_0) = R'(v_0)$.
\end{thm}
We omit the proof of this fundamental theorem since the formal definition of the $\LOCAL$ model has not present here.
However, it has been used in lots of works, in spite of maybe being in different forms.
For example, it serves the same function as \emph{indistinguishability} does in Lynch's textbook \cite{lynch1996distributed}.

\nosection{Distributed OR Algorithm}
Generally speaking, given an underlying graph $G = (V, E)$ in which every node has an initial binary bit, an algorithm $\Pi$ is called \emph{OR algorithm} if some nodes finally could tell whether there exists a node holding ``1'' or not. 
We give the formal definition as follows.
\begin{defn}[OR algorithm] \label{defn:oralgo}
	Let $G = (V, E)$ be an underlying graph, $T$ be a non-empty subset of $V$, $\mathcal{I} = \{0, 1\}$, and $\Pi$ be an algorithm.
	Suppose for arbitrary initial input $I: V \to \mathcal{I}$ and arbitrary random input $R$ if randomness is introduced, every process $v \in T$ would finally terminate with state $q_s$ under execution $\exec{\Pi}{G}{I}$ or $\exec{\Pi}{G}{I,R}$, where $s = \bigvee_{v \in V} I(v)$, within $k_0$ rounds for some $k_0 \in \mathbb{N}$. 
	Then we call $\Pi$ an OR algorithm on underlying graph $G$ with \emph{target set} $T$.
\end{defn}
When discussing OR algorithm, we often refer to some particular initial input $I : V \to \{0,1\}$. 
For convenience we define them as follows.
\begin{eqnarray*}
	I_0 & : & v \mapsto 0 \\
	I_u & : & v \mapsto \bs{1}\{u = v\}
\end{eqnarray*}

\subsection{Threat Model and Anonymity Metric}
Suppose there is an OR algorithm $\Pi$ on underlying graph $G = (V, E)$ with target set $T$, and an execution $\bs{\execs} = \exec{\Pi}{G}{\bs{I}}$ or $\exec{\Pi}{G}{\bs{I}, \bs{R}}$, where initial input $\bs{I}$ and random input $\bs{R}$ are random variables.
We consider an adversary who 
\begin{itemize}
	\item[a)] has unlimited computational power,
	\item[b)] knows $\Pi, G, T$ but does not know $\bs{\execs}$,
	\item[c)] is sure that there exists a unique node $v \in V$ holding ``1'' called \emph{initiator} (and it is actually true), but does not have further a priori knowledge on $\bs{I}$,
	\item[d)] can eavesdrop on some edges $e \in F \subseteq E$ during execution $\bs{\execs}$.
\end{itemize}
By c) we mean that $\bs{I}$ is uniformly distributed on the set $\{I_v : v \in V\}$. 
By d) we mean that the adversary can observe $\hist{\bs{\execs}}{e}$ for each edge $e \in F$. 
The adversary's ultimate objective is to identify the initiator, i.e., to determine $\bs{I}$. 
However, a complete identification is impossible in most cases.
Even so, the adversary do learn some information about the initiator, or $\bs{I}$. 
Taking a) and b) into account, the adversary's ability of identification is only limited by information theory. 
So we use the entropy of $\bs{I}$ conditioned on the adversary's observation $\bs{O}$, or $\entp{\bs{I} | \bs{O}}$, to define anonymity of the initiator.
For example, if $F = \varnothing$ so that the adversary can not derive anything about $\bs{I}$ from its observation, anonymity of the initiator should be $\entp{\bs{I} | \bs{O}} = \entp{\bs{I}} = \log n$ bits\footnote{$\log n$ = $\log_2 n$.}, where $n = |V|$ is the number of nodes in $G$. 
Since we concern about the leakage of anonymity, we will use the metric
\[
	\minfo{\bs{I}}{\bs{O}} = \entp{\bs{I}} - \entp{\bs{I} | \bs{O}}
\]
to measure the anonymity of a certain algorithm $\Pi$. 
Obviously the less it is, the more anonymity $\Pi$ provided to the initiator in a given underlying graph.
\begin{defn}[Anonymity Leakage] \label{defn:leak}
	Let $\Pi$ be an OR algorithm on underlying graph $G = (V, E)$ with target set $T$, $F$ be a subset of $E$, $\bs{I}$ be uniformly distributed on the set $\{I_v : v \in V\}$, $\bs{\execs} = \exec{\Pi}{G}{\bs{I}}$ or $\exec{\Pi}{G}{\bs{I}, \bs{R}}$, and $\bs{O} = ( \hist{\bs{\execs}}{e} : e \in F )$.\footnote{$(f(x) : x \in X)$ denote a vector in $Y^X$, where $Y$ is the codomain of $f$.}
	The \emph{anonymity leakage} is defined by
	\[
		\leak{\Pi}{G,T}{F} = \minfo{\bs{I}}{\bs{O}}
	\]
	where $\minfos$ is the mutual information symbol.
\end{defn}
If $\Pi$ is deterministic, $I$ can determine $O$ so that $\entp{\bs{O} | \bs{I}} = 0$, we have
\[
	\leak{\Pi}{G,T}{F} = \minfo{\bs{I}}{\bs{O}} = \entp{\bs{O}} - \entp{\bs{O} | \bs{I}} = \entp{\bs{O}}.
\]

%% file: body.tex
\section{Lower Bound of Anonymity Leakage} \label{sec:impossibility}

\subsection{Randomized Case}
The following theorem tells us that target set $T$ is not important in most cases, when considering anonymity leakage.
\begin{thm} \label{thm:targetfree}
	Let $\Pi$ be a randomized OR algorithm on underlying graph $G = (V, E)$ with target set $T$. 
	Then for arbitrary $\varnothing \neq T' \subseteq V$, there exists an OR algorithm $\Pi'$ on $G$ with $T'$ such that
	\[
		\leak{\Pi}{G,T}{F} = \leak{\Pi'}{G,T'}{F}
	\] 
	holds for arbitrary $F \subseteq E$.
\end{thm}
\begin{proof}
	We prove the theorem by constructing $\Pi'$. 
	Firstly, $\Pi'$ acts like $\Pi$ until round $k_0$, where $k_0$ depending on $\Pi$ is defined in Definition \ref{defn:oralgo}.
	Then, every processor $v \in T$ knows the value of $s = \bigvee_{v \in V} I(v)$ since $\Pi$ is an OR algorithm.
	$\Pi'$ lets all processors $v \in T$ start broadcasting the value of $s$ just in round $k_0 + 1$.
	The broadcast is achieved by some deterministic distributed broadcast algorithm included in $\Pi'$.
	Finally, all processors $v \in V$, including those in $T'$, would know the value of $s$ because $G$ is connected.
	So $\Pi'$ is an OR algorithm on $G$ with $T'$.
	
	Now consider the difference between $\hist{\execs}{e}$ and $\hist{\execs'}{e}$, where $\execs = \exec{\Pi}{G}{I, R}$ and $\execs' = \exec{\Pi'}{G}{I, R}$ with arbitrary $I$ and $R$. 
	Note that $\Pi'$ always does the same thing regardless of which value $I$ and $R$ would take.
	Therefore, if someone know one of $\hist{\execs}{e}$ and $\hist{\execs'}{e}$, the another can be inferred, or determined, without any information of $I$ and $R$. 
	As a result, the conclusion also holds between $O$ and $O'$, where $O = (\hist{\execs}{e} : e \in F )$ and $O' = ( \hist{\execs'}{e} : e \in F )$ for arbitrary $F \subseteq E$.
	Thus,
	\[
		\leak{\Pi}{G,T}{F} = \minfo{\bs{I}}{\bs{O}} = \minfo{\bs{I}}{\bs{O}, \bs{O}'} = \minfo{\bs{I}}{\bs{O}'} = \leak{\Pi'}{G,T'}{F},
	\]
	where $\bs{I}$ is uniformly distributed on the set $\{I_v : v \in V\}$.
\end{proof}
Then we arrive at the following theorem, which is elementary but important.
Actually, it could be easily deduced from some most basic fact in communication complexity theory, which was introduced by Yao \cite{yao1979some} in 1979.
However, we will still make a rigid proof because it is more tedious to show that our model is compatible with Yao's.
\begin{thm} \label{thm:k2}
	Let $K_2$ be the graph consisting of two nodes $u$ and $v$, which are connected by a single edge $e$. Then any randomized OR algorithm $\Pi$ on $K_2$ with arbitrary target set $T$ will satisfy that
	\[
		\hist{\execs_u}{e} \neq \hist{\execs_v}{e},
	\]
	where $\execs_u = \exec{\Pi}{K_2}{I_u, R_u}$ and $\execs_v = \exec{\Pi}{K_2}{I_v, R_v}$ with arbitrary $R_u$ and $R_v$.
\end{thm}
\begin{proof}
	Prove by contradiction. 
	Suppose $\hist{\execs_u}{e} = \hist{\execs_v}{e}$.
	Without loss of generality assume $u \in T$.
	Let $\execs_0 = \exec{\Pi}{K_2}{I_0, R_0}$ where $R_0(u) = R_v(u)$ and $R_0(v) = R_u(v)$.
	Note that $I_0(u) = I_v(u) = 0$ and $I_0(v) = I_u(v) = 0$.
	Using Theorem \ref{thm:consistency} by substituting $u$ or $v$ for $v_0$ in it respectively, we get
	\[
		\left\{
		\begin{array}{c}
		\pref{\hist{\execs_v}{v \to u}}{k} = \pref{\hist{\execs_0}{v \to u}}{k} \implies \pref{\hist{\execs_v}{u \to v}}{k+1} = \pref{\hist{\execs_0}{u \to v}}{k+1} \\
		\pref{\hist{\execs_u}{u \to v}}{k} = \pref{\hist{\execs_0}{u \to v}}{k} \implies \pref{\hist{\execs_u}{v \to u}}{k+1} = \pref{\hist{\execs_0}{v \to u}}{k+1}
		\end{array}
		\right.
	\]
	Since $\hist{\execs_u}{e} = \hist{\execs_v}{e}$, by induction on $k$ we have $\hist{\execs_u}{e} = \hist{\execs_0}{e} = \hist{\execs_v}{e}$. 
	However, using Theorem \ref{thm:consistency} again for $u$ with $\execs_v$ and $\execs_0$ we get that the state of process $u$ under $\execs_v$ is identical to that under $\execs_0$ forever.
	By Definition \ref{defn:oralgo} this with $u \in T$ gives us $I_v(u) \vee I_v(v) = I_0(u) \vee I_0(v)$, which implies $1 = 0$.
	Absurd.
\end{proof} 
We can generalize Theorem \ref{thm:k2} by using a well-known technique of reduction, which was possibly first introduced by Lamport et al. \cite{lamport1982byzantine}.
\begin{coro} \label{coro:split}
	Let $\Pi$ be a randomized OR algorithm on underlying graph $G = (V, E)$ with target set $T$ and $u, v$ be two nodes in $G$. 
	Suppose that there is a set of edges $F \subseteq E$ whose removal disconnects $u$ from $v$.
	Then, 
	\[
		O^u_F = (\hist{\execs_u}{e} : e \in F) \neq (\hist{\execs_v}{e} : e \in F) = O^v_F,
	\]
	where $\execs_u = \exec{\Pi}{G}{I_u, R_u}$ and $\execs_v = \exec{\Pi}{G}{I_v, R_v}$ with arbitrary $R_u$ and $R_v$.
\end{coro}
\begin{proof}
	Suppose for the sake of contradiction that $O^u_F = O^v_F$.
	We will construct an deterministic OR algorithm $\Pi'$ on $K_2$ violating Theorem \ref{thm:k2}.
	Define $V_u$ as the collection of nodes which are in the same connected component with $u$ in graph $(V, E-F)$.
	And $V_v$ is defined in a similar way.
	Without loss of generality we can assume that there exists a node $t \in T$ not included in $V_v$.
	Let $x$ and $y$ be the two nodes in $K_2$.
	We construct $\Pi'$ as follows.
	
	In general, $\Pi'$ under execution $\exec{\Pi'}{K_2}{I'}$ let processes $x$ and $y$ keep simulating the execution $\exec{\Pi}{G}{I,R}$, step by step.
	To be specific, in round $k$, $x$ keep track of the states of all the processes in $V - V_v$ in the $k$-th round of $\exec{\Pi}{G}{I,R}$, as well as the messages among them.
	Similarly, $y$ is responsible for $V_v$. They assign the initial input and random input by 
	\begin{eqnarray*}
		& \forall w \in V, & I(w) = 
		\left\{
			\begin{array}{ll}
				I'(x), & w = u \\
				I'(y), & w = v \\
				0, & \text{else}
			\end{array}
		\right. \\
		& \forall w \in V - V_v, & R(w) =
		\left\{
			\begin{array}{ll}
				R_u(w), & \text{if } I'(x) = 1 \\
				R_v(w), & \text{if } I'(x) = 0
			\end{array}
		\right. \\
		& \forall w \in V_v, & R(w) =
		\left\{
			\begin{array}{ll}
				R_v(w), & \text{if } I'(y) = 1 \\
				R_u(w), & \text{if } I'(y) = 0
			\end{array}
		\right. 
	\end{eqnarray*}
	which obviously can be achieved without information exchange between $x$ and $y$.
	The remaining question is to simulate those messages between a process in $V-V_v$ and one in $V_v$.
	Actually this part will be the only use of the edge $(x,y)$ in $K_2$ by $\Pi'$.
	Those simulated messages generated in a round under $\exec{\Pi}{G}{I,R}$ will be sent through the edge by some deterministic coding at once.
	The coding should guarantee that the reciver can obtain enough information from it to continue simulating.
	When $x$ finds that simulated $t$ terminates with state $q_s$, it also terminates with the same state.
	
	It is easy to see that $\Pi'$ is an OR algorithm with target set $\{x\}$, so we will not do rigid reasoning here.
	Since $V - V_v$ and $V_v$ are disconnected in graph $(V, E-F)$, by definition of $\Pi'$, $\hist{\exec{\Pi'}{K_2}{I'}}{(x,y)}$ can be determined by $(\hist{\exec{\Pi}{G}{I, R}}{e} : e \in F)$.
	Furthermore, note that processes in $\exec{\Pi'}{K_2}{I_x}$ and $\exec{\Pi'}{K_2}{I_y}$ simulates $\exec{\Pi}{G}{I_u, R_u}$ and $\exec{\Pi}{G}{I_v, R_v}$, respectively.
	Therefore by supposition $O^u_F = O^v_F$, we conclude $\hist{\exec{\Pi'}{K_2}{I_x}}{(x,y)} = \hist{\exec{\Pi'}{K_2}{I_y}}{(x,y)}$ for $\Pi'$.
	Absurd.
\end{proof}
By Corollary \ref{coro:split}, we obtain our first result on anonymity leakage as follows.
\begin{coro} \label{coro:rcase}
	Let $\Pi$ be a randomized OR algorithm on underlying graph $G = (V, E)$ with target set $T$.
	Then, for arbitrary $F \subseteq E$ we have
	\[
		\leak{\Pi}{G,T}{F} \geq \log n + \sum_{V_i} \frac{n_i}{n} \log \frac{n_i}{n}
	\]
	where $n = |V|$, $n_i = |V_i|$, and $V_i$ runs through node set of every connected component in graph $(V, E-F)$. 
\end{coro}
\begin{proof}
	Note that if we can always determine which $V_i$ $v$ is in from an arbitrary $O = ( \hist{\exec{\Pi}{G}{I_v,R}}{e} : e \in F)$, the inequality will immediately follow by $\leak{\Pi}{G,T}{F} = \minfo{\bs{I}}{\bs{O}} = \entp{\bs{I}} - \entp{\bs{I} | \bs{O}}$ in Definition \ref{defn:leak}.
	Let $\mathcal{O}_{V_i} = \{( \hist{\exec{\Pi}{G}{I_v,R}}{e} : e \in F) : v \in V_i \}$.
	Certainly we have $O \in \bigcup_{V_i} \mathcal{O}_{V_i}$.
	Now we only need to prove that all $\mathcal{O}_{V_i}$ are mutually disjoint.
	Suppose for the sake of contradiction that $\mathcal{O}_{V_i} \cap \mathcal{O}_{V_j} \neq \varnothing$ with some $V_i, V_j$ such that $V_i \neq V_j$.
	We choose an $O \in \mathcal{O}_{V_i} \cap \mathcal{O}_{V_j}$.
	By definition of $\mathcal{O}_{V_i}$ we have
	\[
		( \hist{\exec{\Pi}{G}{I_u,R_u}}{e} : e \in F) = O = ( \hist{\exec{\Pi}{G}{I_v,R_v}}{e} : e \in F)
	\]
	for some $u \in V_i, v \in V_j$, which violates Corollary \ref{coro:split}.
\end{proof}

\subsection{Deterministic Case} \label{subsec:deterministic}
Lemma \ref{lemma:path} is the result of Corollary \ref{coro:split} in the deterministic case, which will be very useful when proving our lower bounds of anonymity leakage in this section.
\begin{lemma} \label{lemma:path}
	Let $\Pi$ be a deterministic OR algorithm on underlying graph $G = (V, E)$ with target set $T$ and $u, v$ be two nodes in $G$.
	Define
	\[
		F = \{ e \in E : \hist{\execs_u}{e} \neq \hist{\execs_v}{e} \},
	\]
	where $\execs_u = \exec{\Pi}{G}{I_u}$ and $\execs_v = \exec{\Pi}{G}{I_v}$. Then, $u$ and $v$ are connected in graph $(V, F)$.
	In particular, $|F| \geq d_G(u, v)$.\footnote{$d_G(u, v)$ denotes the distance between two nodes $u$ and $v$ in $G$.}
\end{lemma}
\begin{proof}
	If $u$ and $v$ are unconnected in $(V, F)$, the removal of $E-F$ in $G$ will disconnect $u$ from $v$. 
	By Corollary \ref{coro:split} we have
	\[
		(\hist{\execs_u}{e} : e \in E-F) \neq (\hist{\execs_v}{e} : e \in E-F)
	\]
	which contradicts with the definition of $F$.
	The $d_G(u, v)$ bound of $|F|$ is a direct inference of the fact that $F$ connects $u$ and $v$.
\end{proof}
\begin{thm} \label{thm:sparse}
	Let $\Pi$ be a deterministic OR algorithm on underlying graph $G = (V, E)$ with target set $T$.
	Denote $n = |V|$ and $m = |E|$.
	The distance between two nodes $u$ and $v$ in $G$ is denoted $d_G(u,v)$.
	Then, the following inequality on anonymity leakage holds for arbitrary $k \in \mathbb{N}$.
	\[
		\frac{1}{m^k} \sum_{e_1,\dots,e_k \in E} \leak{\Pi}{G,T}{\{e_1, \dots, e_k\}} \geq - \log \left( \frac{1}{n^2} \sum_{u, v \in V} \left( 1 - \frac{d_G(u,v)}{m} \right)^k \right)
	\]
	Alternatively, for arbitrary $p \in [0,1]$ we have
	\[
		\expe{\leak{\Pi}{G,T}{\bs{F}}} \geq - \log \left( \frac{1}{n^2} \sum_{u, v \in V} \left( 1 - p \right)^{d_G(u,v)} \right)
	\]
	where distribution of $\bs{F}$ is given by that each edge $e \in E$ have possibility $p$ to be included in $\bs{F}$ independently.
\end{thm}
\begin{proof}
	The key of the proof is that for a random variable $\bs{X}$ we always have
	\[
		\entp{\bs{X}} \geq \renyi{2}{\bs{X}}
	\]
	where $\renyi{2}{\bs{X}}$ is the collision entropy, or R\'{e}nyi entropy of order $2$, defined by
	\[
		\renyi{2}{\bs{X}} = - \log \prob{\bs{X} = \bs{Y}}
	\]
	where $\bs{X}$ and $\bs{Y}$ are i.i.d. This can be easily deduced by Jensen's inequality (see \cite[p. 428]{blahut1987principles}).
	Because of determinism of $\Pi$, by Definition \ref{defn:leak} we have
	\[
		\leak{\Pi}{G,T}{F} = \entp{\bs{O}} \geq \renyi{2}{\bs{O}}
	\]
	where $\bs{I}$ is uniformly distributed on the set $\{I_v : v \in V\}$. 
	With shorthand $\hist{v}{e} = \hist{\exec{\Pi}{G}{I_v}}{e}$ we obtain
	\begin{eqnarray*}
		& & \frac{1}{m^k} \sum_{e_1,\dots,e_k \in E} \leak{\Pi}{G,T}{\{e_1, \dots, e_k\}} \\
		&\geq& \frac{1}{m^k} \sum_{e_1,\dots,e_k \in E} \renyi{2}{\hist{\exec{\Pi}{G}{\bs{I}}}{e_i} : i = 1,\dots,k } \\
		&=& \frac{1}{m^k} \sum_{e_1,\dots,e_k \in E} - \log \left( \frac{1}{n^2} \sum_{u, v \in V} \prod_{i = 1}^k \bs{1}\{ \hist{u}{e_i} = \hist{v}{e_i} \} \right) \\
		&\geq& - \log \left( \frac{1}{m^k} \sum_{e_1,\dots,e_k \in E} \frac{1}{n^2} \sum_{u, v \in V} \prod_{i = 1}^k \bs{1}\{ \hist{u}{e_i} = \hist{v}{e_i} \} \right) \\
		&=& - \log \left( \frac{1}{n^2} \sum_{u, v \in V} \frac{1}{m^k} \sum_{e_1,\dots,e_k \in E} \prod_{i = 1}^k \bs{1}\{ \hist{u}{e_i} = \hist{v}{e_i} \} \right) \\
		&=& - \log \left( \frac{1}{n^2} \sum_{u, v \in V} \left( \frac{1}{m} \sum_{e \in E} \bs{1}\{ \hist{u}{e} = \hist{v}{e} \} \right)^k \right) \\
		&=& - \log \left( \frac{1}{n^2} \sum_{u, v \in V} \left( \frac{1}{m} \left( m - \sum_{e \in E} \bs{1}\{ \hist{u}{e} \neq \hist{v}{e} \} \right) \right)^k \right) \\
		&\geq& - \log \left( \frac{1}{n^2} \sum_{u, v \in V} \left( \frac{1}{m} ( m - d_G(u,v) ) \} \right)^k \right) \\
		&=& - \log \left( \frac{1}{n^2} \sum_{u, v \in V} \left( 1 - \frac{d_G(u,v)}{m} \right)^k \right).
	\end{eqnarray*}
	The first ``$\geq$'' holds by the property of collision entropy we have discussed. 
	The second ``$\geq$'' follows from Jensen's inequality since $-\log(x)$ is convex.
	The third ``$\geq$'' is derived from Lemma \ref{lemma:path}.
	Similarly we also obatin
	\begin{eqnarray*}
		& & \expe{\leak{\Pi}{G,T}{\bs{F}}} \\
		&\geq& \expe{\renyi{2}{\hist{\exec{\Pi}{G}{\bs{I}}}{e} : e \in \bs{F} }} \\
		&=& \sum_{F \subseteq E} p^{|F|} (1-p)^{|E-F|} \cdot - \log \left( \frac{1}{n^2} \sum_{u, v \in V} \prod_{e \in F} \bs{1}\{ \hist{u}{e} = \hist{v}{e} \} \right) \\
		&\geq& - \log \left( \sum_{F \subseteq E} p^{|F|} (1-p)^{|E-F|} \frac{1}{n^2} \sum_{u, v \in V} \prod_{e \in F} \bs{1}\{ \hist{u}{e} = \hist{v}{e} \} \right) \\
		&=& - \log \left( \frac{1}{n^2} \sum_{u, v \in V} \sum_{F \subseteq E} p^{|F|} (1-p)^{|E-F|} \prod_{e \in F} \bs{1}\{ \hist{u}{e} = \hist{v}{e} \} \right) \\
		&=& - \log \left( \frac{1}{n^2} \sum_{u, v \in V} \sum_{F \subseteq E} \prod_{e \in E-F} (1-p) \prod_{e \in F} p \cdot \bs{1}\{ \hist{u}{e} = \hist{v}{e} \}  \right) \\
		&=& - \log \left( \frac{1}{n^2} \sum_{u, v \in V} \prod_{e \in E} \left( 1 - p + p \cdot \bs{1}\{ \hist{u}{e} = \hist{v}{e} \} \right) \right) \\
		&=& - \log \left( \frac{1}{n^2} \sum_{u, v \in V} \prod_{e \in E} \left( 1 - p \cdot \bs{1}\{ \hist{u}{e} \neq \hist{v}{e} \} \right) \right) \\
		&\geq& - \log \left( \frac{1}{n^2} \sum_{u, v \in V} \left( 1 - p \right)^{d_G(u,v)} \right).
	\end{eqnarray*}
\end{proof}
Suppose $\Pi$ is a deterministic OR algorithm on $G$ with $T$.
Then in its underlying graph $G = (V, E)$ every $F \subseteq E$ actually induces an equivalence relation $\sim_F$ on $V$, which can be defined by
\[
	u \sim_F v \text{ iff } O^u_F = O^v_F
\]
where symbol $O^v_F$ stands for $(\hist{ \exec{\Pi}{G}{I_v}}{e} : e \in F)$.
Let $c(F)$ be the number of equivalence classes split by relation $\sim_F$.
We have the following lemma.
\begin{lemma} \label{lemma:dense}
	Let $F$ be an arbitrary subset of $E$.
	\begin{itemize}
		\item[a)] There are at least $n - c(F)$ different $e \in E - F$ such that $c(F \cup \{e\}) \geq c(F) + 1$.
		\item[b)] $\displaystyle \leak{\Pi}{G,T}{F} \geq \frac{c(F)-1}{n-1} \log n$.
	\end{itemize}
\end{lemma}
\begin{proof}
	Denote $\hist{u}{e} = \hist{\exec{\Pi}{G}{I_u}}{e}$.
	Consider the edge set of $G$,
	\begin{eqnarray*}
		& & \bigcup_{\substack{u,v \in V \\ u \sim_F v}} \{ e \in E : \hist{u}{e} \neq \hist{v}{e} \} \\
		&=& \{e \in E : \exists u,v \in V,\ u \sim_F v,\ \hist{u}{e} \neq \hist{v}{e} \} \\
		&=& \{e \in E : \exists u,v \in V,\ u \sim_F v,\ u \nsim_{F \cup \{e\}} v \} 
	\end{eqnarray*}
	On one hand, the right-hand side shows that any of its element $e$, which is obviously not in $F$, gives an equivalence relation $\sim_{F \cup \{e\}}$ splitting $V$ into more than $c(F)$ equivalence classes.
	On the other hand, the left-hand side implies that it connects $u$ to $v$ in $G$ if $u \sim_F v$, by Lemma \ref{lemma:path}. 
	So it includes at least $n-c(F)$ edges, which follows from that it contains a spanning forest consisting of $c(F)$ spanning trees for $c(F)$ equivalence classes, respectively.
	Therefore a) holds. 
	
	For b), suppose $\sim_F$ split $V$ into $k$ equivalence classes $V_1, \dots, V_k$.
	Denote $n = |V|$ and $n_i = |V_i|$ for $i = 1, \dots, k$.
	By Definition \ref{defn:leak} we have
	\[
		\leac{\Pi}{G, T}{F} = \entp{\bs{O}} = - \sum_{i=1}^k \frac{n_i}{n} \log \frac{n_i}{n}.
	\]
	Applying Theorem \ref{thm:petrov} with $m = 1$ we immediately get the result.
\end{proof}
\begin{thm} \label{thm:dense}
	Let $\Pi$ be a deterministic OR algorithm on underlying graph $G = (V, E)$ with target set $T$.
	Denote $n = |V|$ and $m = |E|$.
	Then, the following inequality on anonymity leakage holds for arbitrary $k \in \mathbb{N}$.
	\[
		\frac{1}{m^k} \sum_{e_1,\dots,e_k \in E} \leak{\Pi}{G,T}{\{e_1, \dots, e_k\}} \geq \left( 1 - \left( 1 - \frac{1}{m} \right)^k \right) \log n
	\]
	Alternatively, for arbitrary $p \in [0,1]$ we have
	\[
		\expe{\leak{\Pi}{G,T}{\bs{F}}} \geq p \log n
	\]
	where distribution of $\bs{F}$ is given by that each edge $e \in E$ have possibility $p$ to be included in $\bs{F}$ independently.
\end{thm}
\begin{proof}
	By Lemma \ref{lemma:dense} b), to verify the first inequality we only need to prove
	\[
		\frac{1}{m^k} \sum_{e_1,\dots,e_k \in E} \frac{c(\{e_1, \dots, e_k\})-1}{n-1} \geq 1 - \left( 1 - \frac{1}{m} \right)^k 
	\]
	We do by induction on $k$.
	\begin{itemize}
		\item Basis. 
		It is trivial when $k=0$.
		\item Inductive step.
		Suppose we have an edge set $F \subseteq E$.
		By Lemma \ref{lemma:dense} a),
		\[
			\sum_{e \in E} c(F \cup \{e\}) \geq m \cdot c(F) + (n - c(F)) = (m-1)c(F) + n
		\]
		Combining this and induction hypothesis we can complete the step.
		\begin{eqnarray*}
			& & \frac{1}{m^{k+1}} \sum_{e_1,\dots,e_{k+1} \in E} \frac{c(\{e_1, \dots, e_{k+1}\})-1}{n-1} \\
			&=& \frac{1}{m^{k+1}} \sum_{e_1,\dots,e_k \in E} \sum_{e_{k+1} \in E} \frac{c(\{e_1, \dots, e_k\} \cup \{e_{k+1}\})-1}{n-1} \\
			&\geq& \frac{1}{m^{k+1}} \sum_{e_1,\dots,e_k \in E} \frac{(m-1)c(\{e_1, \dots, e_k\}\})+n-m}{n-1} \\
			&=& \frac{m-1}{m} \frac{1}{m^{k}} \sum_{e_1,\dots,e_k \in E} \frac{c(\{e_1, \dots, e_k\}\})-1}{n-1} + \frac{1}{m^{k+1}} \sum_{e_1,\dots,e_k \in E} \frac{n-1}{n-1} \\
 			&\geq& \frac{m-1}{m} \left( 1 - \left( 1 - \frac{1}{m} \right)^k  \right) + \frac{1}{m} \\
 			&=& 1 - \left( 1 - \frac{1}{m} \right)^{k+1}
		\end{eqnarray*}
	\end{itemize}
	Next we verify the second inequality.
	Similarly, by Lemma \ref{lemma:dense} it is enough if
	\begin{eqnarray*}
		& & \sum_{F \subseteq E} p^{|F|} (1-p)^{|E-F|} \frac{c(F)-1}{n-1} \ge p \\
		&\Longleftarrow& \forall k \in \mathbb{N}, \sum_{\substack{F \subseteq E \\ |F|=k}} \frac{c(F)-1}{n-1} \geq \binom{m-1}{k-1}
	\end{eqnarray*}
	We prove it by induction as well.
	\begin{itemize}
		\item Basis. 
		It is trivial when $k=0$.
		\item Inductive step.
		Likewise, by Lemma \ref{lemma:dense} a),
		\[
			\sum_{e \in E-F} c(F \cup \{e\}) \geq (m-k) c(F) + (n - c(F)) = (m-k-1)c(F) + n
		\]
		Therefore,
		\begin{eqnarray*}
			& & \sum_{\substack{F \subseteq E \\ |F|=k+1}} \frac{c(F)-1}{n-1} \\
			&=& \frac{1}{k+1} \sum_{\substack{F \subseteq E \\ |F|=k}} \sum_{e \in E - F} \frac{c(F \cup \{e\})-1}{n-1} \\
			&\geq& \frac{1}{k+1} \sum_{\substack{F \subseteq E \\ |F|=k}} \frac{(m-k-1)c(F) + n - (m-k)}{n-1} \\
			&=& \frac{m-k-1}{k+1} \sum_{\substack{F \subseteq E \\ |F|=k}} \frac{c(F) - 1}{n-1} + \frac{1}{k+1} \sum_{\substack{F \subseteq E \\ |F|=k}} \frac{n-1}{n-1} \\
			&\geq& \frac{m-k-1}{k+1} \binom{m-1}{k-1} + \frac{1}{k+1} \binom{m}{k} \\
			&=& \binom{m-1}{k}
		\end{eqnarray*}
	\end{itemize}
\end{proof}

\section{Trade-off between Anonymity and Communication Complexity} \label{sec:tradeoff}
The \emph{communication complexity} is typically measured in terms of the total number of non-empty messages that are sent.
Note that we do not consider the number of bits in messages here.
\begin{defn}[Communication Complexity] \label{defn:cc}
	Let $\Pi$ be an algorithm on underlying graph $G = (V, E)$.
	Then the \emph{worst-case communication complexity} is defined by
	\[
		\wcom{G}{\Pi} = \sup_{I, R} \left\{ \sum_{\substack{e \in E}} \card{\hist{\exec{\Pi}{G}{I, R}}{e}} \right\}
	\]
	where $\card{\hist{\execs}{e}} = \card{\hist{\execs}{u \to v}} + \card{\hist{\execs}{v \to u}}$, in which $e = (u,v)$ and $\card{\hist{\execs}{u \to v}}$ stands for the number of non-empty messages in $\hist{\execs}{u \to v} = (m_1, m_2, \dots)$, i.e., $\card{\hist{\execs}{u \to v}} = \sum_{i \in \mathbb{Z}^+} \bs{1}\{m_i \neq \lambda \}$.
	Likewise, the \emph{average-case communication complexity} is defined by
	\[
		\acom{G}{\Pi, \bs{I}} = \expe{ \sum_{\substack{e \in E}} \card{\hist{\exec{\Pi}{G}{\bs{I}, \bs{R}}}{e}} }.
	\]
\end{defn}

\subsection{Deterministic Case} \label{subsec:deterministicc}
It is hard for us to improve the result about anonymity leakage in Section \ref{subsec:deterministic} even if we restrict  communication complexity of the algorithm.
However, we can find a trade-off between anonymity and communication complexity, if we consider a slightly different threat model.
That is, the adversary can \emph{not} distinguish between two non-empty messages.
This model corresponds to the case where communication is protected by a semantically secure cryptosystem in the real world.
One should note that though we are studying the deterministic case in this section, the model works for the randomized case too.
\begin{defn} \label{defn:leac}
	Let $\Pi$ be an OR algorithm on underlying graph $G = (V, E)$ with target set $T$, $F$ be a subset of $E$, $\bs{I}$ be uniformly distributed on the set $\{I_v : v \in V\}$, $\bs{\execs} = \exec{\Pi}{G}{\bs{I}}$ or $\exec{\Pi}{G}{\bs{I}, \bs{R}}$.
	Define $\filt{\hist{\execs}{e}} \in \{0,1\}^{\omega}$ for $e = (u,v)$ by 
	\[
		\filt{\hist{\execs}{e}} = ( \filt{\hist{\execs}{u \to v}}, \filt{\hist{\execs}{v \to u}} )
	\]
	and
	\[
		\filt{\hist{\execs}{u \to v}} = (\bs{1}\{m_1 \neq \lambda \}, \bs{1}\{m_2 \neq \lambda \}, \dots)
	\] 
	with $\hist{\execs}{u \to v} = (m_1, m_2, \dots)$.
	Let $\bs{O} = ( \filt{\hist{\bs{\execs}}{e}} : e \in F )$.
	The \emph{anonymity leakage with binary filter} is defined by
	\[
		\leac{\Pi}{G,T}{F} = \minfo{\bs{I}}{\bs{O}}.
	\]
\end{defn}
\begin{lemma} \label{lemma:card}
	Use the same symbols as in Definition \ref{defn:leac}. Define
	\[
		\bs{O^*} = ( \card{\hist{\exec{\Pi}{G}{\bs{I}, \bs{R}}}{e}} : e \in F ).
	\]
	Then we have $\leac{\Pi}{G,T}{F} \geq \minfo{\bs{I}}{\bs{O^*}}$, or $\entp{\bs{O^*}}$ when $\Pi$ is deterministic.
\end{lemma}
\begin{proof}
	Obviously $\bs{O}$ determines $\bs{O^*}$. Hence
	\[
		\leac{\Pi}{G,T}{F} = \minfo{\bs{I}}{\bs{O}} = \minfo{\bs{I}}{\bs{O}, \bs{O^*}} \geq \minfo{\bs{I}}{\bs{O^*}}.
	\]
	If $\Pi$ is deterministic, $\bs{I}$ determines $\bs{O^*}$, thus $\minfo{\bs{I}}{\bs{O^*}} = \entp{\bs{O^*}}$.
\end{proof}
\begin{lemma} \label{lemma:histc}
		Let $\hists$ and $\hists'$ be histories. $\hists \neq \hists'$ implies $\card{\hists} \neq \card{\hists'}$ or $\card{\hists}, \card{\hists'} > 0$.
\end{lemma}
\begin{proof}
	Trivial.
\end{proof}
\begin{defn}[Active Edges]
	For a certain execution $\execs$ on underlying graph $G = (V, E)$, its active edges is a set defined as the collection of those edges by which at least one (non-empty) message is transmitted under $\execs$.
	\[
		\acte{\execs} = \{ e \in E : \card{\hist{\execs}{e}} > 0 \}
	\]
\end{defn}
Let $G = (V, E)$ be a graph.
To obtain a result like Theorem \ref{thm:sparse}, for $E_0 \subseteq E$ we define $\disc{G}{u,v}{E_0}$ as follows.
\[
	\disc{G}{u,v}{E_0} = \inf_{F \subseteq E}\{ |F - E_0| : F \text{ connects } u \text{ to } v \}
\]
\begin{thm} \label{thm:sparsec}
	Let $\Pi$ be a deterministic OR algorithm on underlying graph $G = (V, E)$ with target set $T$.
	Denote $n = |V|$ and $m = |E|$.
	Suppose $\actes_v = \acte{\exec{\Pi}{G}{I_v}}$ for each $v \in V$.
	Then, the following inequality on anonymity leakage holds for arbitrary $k \in \mathbb{N}$.
	\[
		\frac{1}{m^k} \sum_{e_1,\dots,e_k \in E} \leac{\Pi}{G,T}{\{e_1, \dots, e_k\}} \geq - \log \left( \frac{1}{n^2} \sum_{u, v \in V} \left( 1 - \frac{\disc{G}{u,v}{\actes_u \cap \actes_v}}{m} \right)^k \right)
	\]
	Alternatively, for arbitrary $p \in [0,1]$ we have
	\[
		\expe{\leac{\Pi}{G,T}{\bs{F}}} \geq - \log \left( \frac{1}{n^2} \sum_{u, v \in V} \left( 1 - p \right)^{\disc{G}{u,v}{\actes_u \cap \actes_v}} \right)
	\]
	where distribution of $\bs{F}$ is given by that each edge $e \in E$ have possibility $p$ to be included in $\bs{F}$ independently.
\end{thm}
\begin{proof}
	By Lemma \ref{lemma:card}, $\leac{\Pi}{G,T}{F} \geq \entp{\bs{O^*}} \geq \renyi{2}{\bs{O^*}}$.
	Denote $\hist{v}{e} = \hist{\exec{\Pi}{G}{I_v}}{e}$.
	We only need to show that
	\[
		\card{ \{ e \in E: \card{\hist{u}{e}} \neq \card{\hist{v}{e}} \} } \geq \disc{G}{u,v}{\actes_u \cap \actes_v}
	\]
	because the remaining work is almost the same as in proving Theorem \ref{thm:sparse}.
	By Lemma \ref{lemma:histc} we have
	\begin{eqnarray*}
		\{ e \in E: \card{\hist{u}{e}} \neq \card{\hist{v}{e}} \} &\supseteq& \{ e \in E: \hist{u}{e} \neq \hist{v}{e} \} - \{ e \in E: \card{\hist{u}{e}}, \card{\hist{v}{e}} > 0 \} \\
		&=& \{ e \in E: \hist{u}{e} \neq \hist{v}{e} \} - (\actes_u \cap \actes_v)
	\end{eqnarray*}
	Comparing this with Lemma \ref{lemma:path} and definition of $\discs$ we reach the goal.
\end{proof}
\begin{coro} \label{coro:sparsec}
	Let $\Pi$ be a deterministic OR algorithm on underlying graph $G = (V, E)$ with target set $T$.
	Denote $n = |V|$ and $m = |E|$ and suppose $\wcoms = \wcom{G}{\Pi}$.
	Let $G[U]$ denote the subgraph induced in $G$ by $U$ and $d_G(U,v) = \inf_{u \in U} \{ d_G(u, v)\}$.
	Then, the following inequality on anonymity leakage holds for arbitrary $k \in \mathbb{N}$.
	\[
		\frac{1}{m^k} \sum_{e_1,\dots,e_k \in E} \leac{\Pi}{G,T}{\{e_1, \dots, e_k\}} \geq - \log \left( \frac{1}{n^2} \sum_{u \in V} \inf_{\substack{u \in U \subseteq V \\ |U| \leq \wcoms + 1 \\ G[U] \mathrm{\; connected}}} \left\{ \sum_{v \in V} \left( 1 - \frac{d_G(U,v)}{m} \right)^k \right\} \right)
	\]
	Alternatively, for arbitrary $p \in [0,1]$ we have
	\[
		\expe{\leac{\Pi}{G,T}{\bs{F}}} \geq - \log \left( \frac{1}{n^2} \sum_{u \in V} \inf_{\substack{u \in U \subseteq V \\ |U| \leq \wcoms + 1 \\ G[U] \mathrm{\; connected}}} \left\{ \sum_{v \in V} \left( 1 - p \right)^{d_G(U,v)} \right\} \right)
	\]
	where distribution of $\bs{F}$ is given by that each edge $e \in E$ have possibility $p$ to be included in $\bs{F}$ independently.
\end{coro}
Suppose $\Pi$ is a deterministic OR algorithm on $G$ with $T$.
Like what we do in Section \ref{subsec:deterministic}, we define an equivalence relation $\sim_F^*$ on $V$ by
\[
	u \sim^*_F v \text{ iff } {O^*}^u_F = {O^*}^v_F
\]
where ${O^*}^v_F$ stands for $( \card{\hist{\exec{\Pi}{G}{I_v}}{e}} : e \in F )$.
Suppose there are $k$ equivalence classes $V_1, \dots, V_k$ under $\sim^*_F$.
Denote $n = |V|$ and $n_i = |V_i|$ and suppose $\wcoms = \wcom{G}{\Pi} < n - 1$.
With definition
\[
	\chi^*(F) = \sum_{i=1}^k \min \{ \wcoms + 1, n_i \}
\]
we have the following lemma.
\begin{lemma} \label{lemma:densec}
	Let $F$ be an arbitrary subset of $E$.
	\begin{itemize}
		\item[a)] There are at least $n - \chi^*(F)$ different $e \in E - F$ such that $\chi^*(F \cup \{e\}) \geq \chi^*(F) + 1$.
		\item[b)] $\displaystyle \leac{\Pi}{G,T}{F} \geq \frac{\chi^*(F)-(\wcoms + 1)}{n-(\wcoms + 1)} \log \frac{n}{\wcoms + 1}$.
	\end{itemize}
\end{lemma}
\begin{proof}
	For any equivalence class $V_i$ with $n_i > \wcoms + 1$ we arbitrarily choose a node $u \in V_i$.
	Denote $\hist{v}{e} = \hist{\exec{\Pi}{G}{I_v}}{e}$.
	By Lemma \ref{lemma:histc} we have
	\begin{eqnarray*}
		& & \{ e \in E : \exists v \sim^*_F u, u \nsim^*_{F \cup \{e\}} v \} \\
		&=& \{ e \in E : \exists v \sim^*_F u,\card{\hist{u}{e}} \neq \card{\hist{v}{e}} \} \\
		&\supseteq& \{ e \in E : \exists v \sim^*_F u,\hist{u}{e} \neq \hist{v}{e} \} - \{ e \in E : \exists v \sim^*_F u,\card{\hist{u}{e}},\card{\hist{v}{e}}>0 \} \\
		&\supseteq& \bigcup_{\substack{v \in V \\ v \sim^*_F u}} \{ e \in E : \hist{u}{e} \neq \hist{v}{e} \} - \acte{\exec{\Pi}{G}{I_u}}
	\end{eqnarray*}
	By Lemma \ref{lemma:path} the first term in the right-hand side has at least $n_i - 1$ edges since it connects $V_i$.
	And certainly $|\acte{\exec{\Pi}{G}{I_u}}| \leq \wcoms$.
	So the left-hand side includes at least $n_i - \wcoms - 1$ edges.
	We show that each edge $e$ in it satisfies $\chi^*(F \cup \{e\}) > \chi^*(F)$, by rewriting
	\[
		\chi^*(F) = \sum_{i=1}^k \min \left\{ \wcoms + 1, n_i \right\} = \sum_{\substack{v \in V}} \min \left\{ \frac{\wcoms + 1}{|[v]_{\sim^*_F}|}, 1 \right\}
	\]
	where $[v]_{\sim}$ stands for the equivalence class under relation $\sim$ including $u$.
	Because each $V_i$ for $i = 1,\dots,n$ gives us at least $\max\{0, n_i - \wcoms - 1\}$ edges we want, summing them up we have at least $n - \chi^*(F)$ edges.
	Therefore a) holds.
	
	For b), by Lemma \ref{lemma:card} we have
	\[
		\leac{\Pi}{G, T}{F} \geq \entp{\bs{O^*}} = - \sum_{i=1}^k \frac{n_i}{n} \log \frac{n_i}{n}.
	\]
	Applying Theorem \ref{thm:petrov} with $m = \wcoms + 1$ we immediately get the result.
\end{proof}
\begin{thm} \label{thm:densec}
	Let $\Pi$ be a deterministic OR algorithm on underlying graph $G = (V, E)$ with target set $T$.
	Denote $n = |V|$ and $m = |E|$ and suppose $\wcoms = \wcom{G}{\Pi} < n - 1$.
	Then, the following inequality on anonymity leakage holds for arbitrary $k \in \mathbb{N}$.
	\[
		\frac{1}{m^k} \sum_{e_1,\dots,e_k \in E} \leac{\Pi}{G,T}{\{e_1, \dots, e_k\}} \geq \left( 1 - \left( 1 - \frac{1}{m} \right)^k \right) \log \frac{n}{\wcoms + 1}
	\]
	Alternatively, for arbitrary $p \in [0,1]$ we have
	\[
		\expe{\leac{\Pi}{G,T}{\bs{F}}} \geq p \log \frac{n}{\wcoms + 1}
	\]
	where distribution of $\bs{F}$ is given by that each edge $e \in E$ have possibility $p$ to be included in $\bs{F}$ independently.
\end{thm}
\begin{proof}
	The proof is almost the same as that for Theorem \ref{thm:dense}, except using Lemma \ref{lemma:densec} instead of Lemma \ref{lemma:dense}.
\end{proof}

\subsection{Difficulty in Studying Randomized Case}
In this section, we try to find out the essential difficulty in studying the randomized case when using Corollary \ref{coro:split} as relaxation.
Consider a simplest underlying graph $G = (V, E)$, which includes only two nodes $u,v$ and $n$ edges $e_1, \dots, e_n$ connecting them.
Given a certain randomized OR algorithm $\Pi$ on $G$, we denote $\bs{X_i} = \hist{\exec{\Pi}{G}{\bs{I}, \bs{R}}}{e_i}$ for $i = 1, \dots, n$, where $\bs{I}$ is uniformly distributed on $\{I_u, I_v\}$.
By plugging $F = E$ into Corollary \ref{coro:split} we obtain that $\bs{X_1}, \dots, \bs{X_n}$ determine $\bs{I}$.
If the adversary can observe randomly $k$ of $\bs{X_1}, \dots, \bs{X_n}$, the expectation of anonymity leakage is
\[
	\epsilon = \frac{1}{\binom{n}{k}}\sum_{\substack{F \subseteq E \\ |F| = k}} \leak{\Pi}{G,T}{F} = \frac{1}{\binom{n}{k}} \sum_{\substack{J \subseteq \{1,\dots,n\} \\ |J|=k}} \minfo{\bs{I}}{(\bs{X_i} : i \in J)}
\]
There does exist a trade-off between $\epsilon$ and $\sum_{i=1}^n \entp{\bs{X_i}}$ which clearly bounds the average-case communication compleity $\acom{G}{\Pi, \bs{I}}$.
\begin{eqnarray*}
	& & \frac{1}{n}\sum_{i=1}^n \entp{\bs{X_i}} \\
	&\geq& \frac{1}{(n-k)\binom{n}{k}} \sum_{\substack{J \subseteq \{1,\dots,n\} \\ |J|=k}} \entp{(\bs{X_i}:i \notin J)} \\
	&\geq& \frac{1}{(n-k)\binom{n}{k}}  \sum_{\substack{J \subseteq \{1,\dots,n\} \\ |J|=k}} \minfo{\bs{I}}{(\bs{X_i}:i \notin J) | (\bs{X_i}:i \in J)} \\
	&=& \frac{1}{(n-k)\binom{n}{k}}  \sum_{\substack{J \subseteq \{1,\dots,n\} \\ |J|=k}} 
	\left( \minfo{\bs{I}}{\bs{X_1},\dots,\bs{X_n}} - \minfo{\bs{I}}{(\bs{X_i}:i \in J)} \right) \\
	&=& \frac{1}{n-k}(\entp{\bs{I}}-\epsilon) = \frac{1-\epsilon}{n-k}
\end{eqnarray*}
However we hardly consider it as a meaningful result, since we can only obtain $\sum_{i=1}^n \entp{\bs{X_i}} \geq 2$ with a relatively strong condition $k = n/2$ and $\epsilon = 0$.
On the other hand, we have not found a construction of $\bs{X_1}, \dots, \bs{X_n}$ to approach the bound given above even within a constant factor.
In fact, if $\epsilon = 0$, we arrive at a special case called \emph{$(k,n,n)$-ramp scheme} in the study of \emph{secret sharing scheme}.
Next we will briefly introduce some necessary concepts about secret sharing scheme, which are mostly taken from \cite{bogdanov2016threshold} with some modification.

Let $\mathcal{P}$ be a set of $n$ parties, i.e., $|\mathcal{P}| = n$.
A collection of subsets $\mathcal{A} \subseteq 2^{\mathcal{P}}$ is \emph{monotone} if for every $B \in \mathcal{A}$ and $B \subseteq C$ it holds that $C \in \mathcal{A}$.
The collection is \emph{anti-monotone} if for every $B \in \mathcal{A}$ and $C \subseteq B$ it holds that $C \in \mathcal{A}$.
\begin{defn}
	A \emph{partial access structure} $\mathcal{A} = (\mathcal{S}, \mathcal{R})$ is a pair of non-empty disjoint collections of subsets $\mathcal{R}$ and $\mathcal{S}$ of $2^{\mathcal{P}}$ such that $\mathcal{R}$ is monotone and $\mathcal{S}$ is anti-monotone.
	Subsets in $\mathcal{R}$ are call \emph{qualified} and subsets in $\mathcal{S}$ are called \emph{unqualified}.
\end{defn}
A secret sharing scheme for $\mathcal{A} = (\mathcal{S}, \mathcal{R})$ is a method by which a dealer who has a secret distributes shares to $n$ parties such that any subset in $\mathcal{R}$ can reconstruct the secret from its shares, while any subset in $\mathcal{S}$ can not reveal any information on the secret.
\begin{defn}
	A secret sharing scheme of a secret $\bs{S}$ for a partial access structure $\mathcal{A} = (\mathcal{S}, \mathcal{R})$ over $n$ parties is the collection of random variables $\bs{S_1}, \dots, \bs{S_n}$ with the following properties: 
	\begin{description}
		\item \textbf{Reconstruction.} For every $B \in \mathcal{R}$, $\entp{\bs{S} | (\bs{S_i} : i \in B)} = 0$.
		\item \textbf{Secrecy.} For every $B \in \mathcal{S}$, $\entp{\bs{S} | (\bs{S_i} : i \in B)} = \entp{\bs{S}}$.
	\end{description}
	In particular, we say that a scheme is $(s, r, n)$-ramp scheme if its partial access structure $\mathcal{A} = (\mathcal{S}, \mathcal{R})$ is given by $\mathcal{S} = \{B : |B| \leq s\}$ and $\mathcal{R} = \{B : |B| \geq r\}$, where $s,r,n$ are all integers and $0 \leq s < r \leq n$.
\end{defn}
Therefore, by the definition above, when $\epsilon = 0$, it is easy to see that $\bs{X_1}, \dots, \bs{X_n}$ form a $(k,n,n)$-ramp scheme for secret $\bs{I}$.
Recall the bound we have obtained
\[
	\sum_{i=1}^{n}\entp{\bs{X_i}} \geq \frac{n}{n-k} \entp{\bs{I}} \tag{$\ast$}
\]
We will soon see that this bound can not be improved by using only Shannon-type inequalities, i.e., $\minfo{\bs{X}}{\bs{Y}|\bs{Z}} \geq 0$.

In the context of secret sharing scheme, we call a $(s, r, n)$-ramp scheme \emph{optimal} if it has the property that $\entp{\bs{S_i}} = \entp{\bs{S}}/(r-s)$ for $i = 1, \dots, n$.
In fact, for arbitrary $(s, r, n)$, an optimal $(s, r, n)$-ramp scheme always exists \cite[Theorem 9]{jackson1996combinatorial}, which means $(\ast)$ is tight if we can allow $\bs{I},\bs{X_1}, \dots, \bs{X_n}$ to follow some certain distributions, regardless of restrictions from their definition.
However we can not, since at least it is clear that $\bs{I}$ is a $1$-bit binary secret.
To our best knowledge, there is still no non-trivial result for this case to some extent.
A most relevant result for $1$-bit binary secret is that recently, Bogdanov et al. \cite{bogdanov2016threshold} and Cascudo et al. \cite{cascudo2013bounds} prove that for $1 \leq s < r < n$, any $(s, r, n)$-ramp secret sharing scheme must have \emph{share size} at least
\[
	\max \left\{ \log \frac{n-s+1}{r-s}, \log \frac{r+1}{r-s} \right\}
\]
where share size is defined as $\max_{i=1,\dots,n} \{ \renyi{0}{\bs{S_i}} \}$, in which $\renyi{0}{\bs{X}}$ is the  Hartley entropy, or R\'{e}nyi entropy of order $0$, given by the logarithm of the size of the support of $\bs{X}$.

More generally, consider an OR algorithm $\Pi$ on an arbitrary underlying graph $G = (V, E)$.
Suppose $\bs{I}$ is still uniformly distributed on $\{I_u, I_v\}$ with some chosen $u, v \in V$.
Then, by Corollary \ref{coro:split}, all $\bs{S_e} = \hist{\exec{\Pi}{G}{\bs{I}, \bs{R}}}{e}$ for $e \in E$ form a secret sharing scheme for secret $\bs{I}$ with access structure $\mathcal{A} = (\mathcal{S}, \mathcal{R})$, where $\mathcal{R} \subseteq 2^{E}$ consists of all subsets of $E$ whose removal disconnects $u$ from $v$ in $G$, and $\mathcal{S}$ is not necessarily specified.
An interesting thing is that, in a sense, this access structure looks like a dual structure of the \emph{undirected s-t-connectivity structure} \cite[Section 3.2]{beimel2011secret}.\footnote{Of course, actually they have no mathematical relationship.}

%% file: discussion.tex
\section{Discussion}

\subsection{Tightness Issue on Our Bounds}
For the deterministic case, the canonical \emph{convergecast} algorithm \cite{peleg2000distributed} can be modified to be an OR algorithm with a good performance.
To be specific, we can choose an undirected spanning tree $Y$ of underlying graph $G = (V, E)$.
Let $t$ be chosen as a ``root'' of $Y$.
Denote the depth of a node $v$ in tree $Y$ as $\dept{Y}{v} = d_Y(v,t)$.
The parent of a node is the node connected to it on the path to the root $t$.
Then, it is elementary to see that Algorithm \ref{algo:convergecast} is an OR algorithm on $G$ with target set $T = \{t\}$, where $n = |V|$.
\begin{algorithm}[H] \label{algo:convergecast}
	\caption{Convergecast}
	\begin{algorithmic}[1]
		\item // code for node $v$ with input $i = I(v) \in \{0,1\}$ 
		\State {\bf initially} $x \leftarrow i$
		\Upon {receiving a message $\langle m \rangle$ from any neighbor}
			\State $x \leftarrow x \vee m$
		\EndUpon
		\Upon {it is round $n - \dept{Y}{v}$}
			\If {$v = t$}
				\State \Return $x$
			\Else
				\State send a message  $\langle x \rangle$ to $v$'s parent
				\State \Return
			\EndIf
		\EndUpon
	\end{algorithmic}  
\end{algorithm} 
Let $\Pi_0(Y)$ represent the algorithm.
Because of its simple behavior, the performance can be determined easily as follows.\footnote{We omit the proof since it is easy and not relatively important.}
\[
	\left\{
		\begin{array}{l}
			\displaystyle 
			\leak{\Pi_0(Y)}{G,T}{F} = \log n + \sum_{V_i} \frac{n_i}{n} \log \frac{n_i}{n} \\
			\leac{\Pi_0(Y)}{G,T}{F} = 0 \\
			\wcom{G}{\Pi_0(Y)} = n - 1
		\end{array}
	\right.
\]
where $n = |V|$, $n_i = |V_i|$, and $V_i$ runs through node set of every connected component in graph $(V, E(Y)-F)$, in which $E(Y)$ is determined by $Y = (V, E(Y))$.

By applying $\Pi_0(Y)$ on underlying graph $K_n$, the complete graph of order $n$, where $Y$ consists of one internal node $t$ and $n-1$ leaves, we immediately get that the bound given by Theorem \ref{thm:dense} is asymptotically tight.
Theorem \ref{thm:sparse} is certainly not asymptotically tight for dense underlying graph since in this case the bound it gives is rather weaker than that from Theorem \ref{thm:dense}.
However, it seems like that on many sparse underlying graphs, the anonymity leakage for $\Pi_0(Y)$ is not far away from the bound given by Theorem \ref{thm:sparse}, if we choose an appropriate $Y$.
As for Theorem \ref{thm:sparsec} an Theorem \ref{thm:densec}, we know little about tightness concerns.

For the randomized case, an interesting problem is that, can we design an randomized OR algorithm $\Pi$ on a given underlying graph $G$, which always reachs the bound given by Theorem \ref{coro:rcase} for arbitrary $F$?
If the answer is ``yes'', the bound is absolutely tight in every sense.

\subsection{Other Issues} \label{subsec:otherissues}

\nosection{Upper Bound of Anonymity Leakage}
In this paper, we give a series of lower bounds of anonymity leakage, while we know little about the upper bounds.
At present, the best upper bound for the deterministic case we know is given by $\Pi_0(Y)$ introduced in the previous section.
Does there exist a better general deterministic OR algorithm, or a better one for some certain underlying graph?
Also, it is meaningful to find a good algorithm for the randomized case, or with threat model defined in Section \ref{subsec:deterministicc}.

\nosection{Limitation of Our Model}
Undoubtedly, the simplicity of our model brings us a lot of advantages, but we lose something at the same time.
As we mentioned at the beginning of this paper, our model is not well compatible with cryptography.
Also, a real-world anonymous communication system may change its topology even during a transmission, e.g., a new edge caused by a query for another node's IP address.
Such a function apparently could not be implemented within our model.

\nosection{Flexibility of Our Model}
Throughout this paper, we only consider an external passive adversary whose ability is to observe some edges.
However, we could get parallel results when the adversary is still passive but internal (i.e., its ability is to observe some processors) with almost the same method we have used.
If we want to study the case an adversary is internal and active, which means that there are some nodes in the control of the adversary, then, it is equivalent to say that the adversary's ability is to hold some Byzantine nodes \cite{lamport1982byzantine} in our model.

%% file: appendix.tex
\section{Appendix}
The proof of the following inequality is attributed to Fedor Petrov.
\begin{thm} \label{thm:petrov}
	Let $n, m, n_1, n_2, \dots, n_k \in \mathbb{Z}^+$ such that $m < n = \sum_{i=1}^k n_i$. Then,
	$$
		-\sum_{i=1}^k \frac{n_i}{n} \log \frac{n_i}{n} \geq \frac{\sum_{i=1}^k \min \{n_i, m\} - m}{n - m} \log \frac{n}{m}.
	$$
\end{thm}
\begin{proof}
	(modified from \cite{282344})
	Assume that $n_1,\dots,n_t<m\leqslant n_{t+1},\dots,n_k$. Denote $p_i=n_i/n$, $a=m/n$. Let $H(p)=-p\log p$ be the entropy function. What we want to prove is 
	\[
		\sum_{i=1}^k H(p_i)\geq \frac{p_1+\dots+p_t+a(k-t-1)}{1-a}\log a^{-1}.
	\]
	Note that $H(p)$ is concave function, thus we have $H(p_i)\geqslant H(a)\cdot (p_i/a)$ for $i\leq t$ and $H(p_i)\geq H(a)\cdot (1-p_i)/(1-a)$ for $i>t$. Substituting these estimates for $H(p_i)$ to the left-hand site and simplifying we get exactly right-hand side.
\end{proof}